\documentclass[twocolumn,10pt]{IEEEtran}

\usepackage{amsmath,amssymb,amsthm}
\usepackage{graphicx}
\usepackage{xcolor}
\usepackage{subfigure}
\usepackage{multirow}
\usepackage{comment}
\usepackage[shortlabels]{enumitem}
\usepackage[acronym]{glossaries}

\newtheorem{proposition}{Proposition}

\newacronym{mimo}{MIMO}{multiple-input multiple-output}
\newacronym{rf}{RF}{radio frequency}
\newacronym{awgn}{AWGN}{additive white Gaussian noise}
\newacronym{xpd}{XPD}{cross-polar discrimination}
\newacronym{iid}{i.i.d.}{independent and identically distributed}
\newacronym{los}{LoS}{line-of-sight}
\newacronym{tx}{Tx}{transmitter}
\newacronym{rx}{Rx}{receiver}

\begin{document}

\title{\huge Dual-Polarized Beyond Diagonal RIS}

\author{Matteo~Nerini,~\IEEEmembership{Member,~IEEE},
        Bruno~Clerckx,~\IEEEmembership{Fellow,~IEEE}

\thanks{This work has been supported in part by UKRI under Grant EP/Y004086/1, EP/X040569/1, EP/Y037197/1, EP/X04047X/1, EP/Y037243/1.}
\thanks{Matteo Nerini and Bruno Clerckx are with the Department of Electrical and Electronic Engineering, Imperial College London, SW7 2AZ London, U.K. (e-mail: m.nerini20@imperial.ac.uk; b.clerckx@imperial.ac.uk).}}

\maketitle

\begin{abstract}
Beyond diagonal reconfigurable intelligent surface (BD-RIS) is a family of RIS architectures more flexible than conventional RIS.
While BD-RIS has been primarily analyzed assuming uni-polarized systems, modern wireless deployments are dual-polarized.
To address this gap, this paper investigates the fundamental limits of dual-polarized BD-RIS-aided systems.
We derive the scaling laws governing the performance of BD-RIS and the Pareto frontier of the trade-off between performance and circuit complexity enabled by BD-RIS.
Theoretical results show that the group-connected RIS with group size 2 provides remarkable gains over conventional RIS in both Rayleigh and line-of-sight (LoS) channels, while maintaining a reduced circuit complexity.
\end{abstract}


\section{Introduction}

Reconfigurable intelligent surface (RIS) is expected to play a key role in future wireless communications, allowing us to optimize the propagation environment through surfaces made of multiple elements with reconfigurable scattering properties \cite{wu21}.
Conventionally, a RIS has been implemented by independently controlling each element with a tunable load, leading to a diagonal phase shift matrix, formally known as scattering matrix.
Because of the structure of its scattering matrix, this conventional RIS architecture is referred to as diagonal RIS (D-RIS).
To improve the flexibility of D-RIS, more general RIS architectures have been proposed, under the name of beyond diagonal RIS (BD-RIS), which are characterized by scattering matrices not limited to being diagonal \cite{she22}.

Several BD-RIS architectures have been developed.
The fully-connected RIS has been proposed by interconnecting all the RIS elements to each other through tunable impedance components \cite{she22}.
This architecture has the maximum circuit complexity, in terms of the number of impedance components, and achieves the maximum performance \cite{she22}.
To balance performance and complexity, the group-connected RIS has been developed by dividing the RIS elements into groups and interconnecting only RIS elements within groups \cite{she22}.
In \cite{ner24-2}, the tree- and forest-connected RISs have been proposed, which are the least complex BD-RIS architectures achieving the same performance as fully- and group-connected RISs, respectively, in single-user systems.
The fundamental limit of the trade-off between performance and complexity offered by BD-RIS has been investigated in \cite{ner23}, where its Pareto frontier was analytically provided for single-user systems.
Recently, the Q-stem connected RIS has emerged to provide a favorable performance-complexity trade-off in multi-user systems \cite{zho24}.

The performance offered by BD-RIS has been analyzed in existing works by implicitly assuming uni-polarized systems.
However, modern \gls{mimo} wireless systems utilize dual-polarized antenna arrays to accommodate more antennas within limited space, and improve diversity using the polarization dimension \cite{kim10}.
Dual-polarized RIS-aided systems have been analyzed considering the conventional D-RIS architecture and compared to uni-polarized systems \cite{han22,mun24,zhe24}.
In addition, dual-polarized D-RIS has been proposed to perform broad beamforming \cite{ram23}, to enable holographic \gls{mimo} \cite{zen24}, and to implement massive \gls{mimo} transmission \cite{che21}.
Given the practical importance of dual-polarized RIS, in this study, we analyze the fundamental performance limits of dual-polarized BD-RIS.

The contributions of this study are as follows.
\textit{First}, we analyze the performance achieved by BD-RIS in dual-polarized RIS-aided systems.
Specifically, we derive the scaling laws of the received power considering Rayleigh and \gls{los} channels, and compare them with the scaling laws for uni-polarized systems.
\textit{Second}, we provide the gain of BD-RIS over D-RIS as a function of the inverse of the \gls{xpd}.
Although BD-RIS was known to provide no gain over D-RIS with LoS channels \cite{she22}, we analytically show that BD-RIS offers remarkable gains over D-RIS in dual-polarized systems with both Rayleigh and \gls{los} channels.
\textit{Third}, we analytically derive the Pareto frontier of the trade-off between performance and circuit complexity offered by BD-RIS in dual-polarized systems with \gls{los} channels, which differs from the one derived for uni-polarized systems in \cite{ner23}.
Although the tree-connected RIS achieves maximum performance under any channel realization with
the lowest circuit complexity, we prove that a lower-complexity BD-RIS is sufficient to get maximum performance in dual-polarized LoS systems, i.e., the group-connected RIS with group size 2.

\section{Dual-Polarized RIS-Aided System Model}

Consider a wireless system between a single-antenna \gls{tx} and a single-antenna \gls{rx}, e.g., an access point and a mobile terminal in cell-free \gls{mimo}, aided by a RIS with $N$ elements.
Denoting the transmitted signal as $x\in\mathbb{C}$, the received signal $y\in\mathbb{C}$ is given by $y=hx+n$, where $h\in\mathbb{C}$ is the wireless channel between \gls{tx} and \gls{rx} and $n\in\mathbb{C}$ is the noise.
The wireless channel $h$ is modeled as a function of the RIS scattering matrix $\boldsymbol{\Theta}\in\mathbb{C}^{N\times N}$ according to the widely used cascaded model $h=\mathbf{h}_{R}\boldsymbol{\Theta}\mathbf{h}_{T}$, where $\mathbf{h}_{R}\in\mathbb{C}^{1\times N}$ is the channel between RIS and \gls{rx}, $\mathbf{h}_{T}\in\mathbb{C}^{N\times 1}$ is the channel between \gls{tx} and RIS, and we neglect the direct channel between \gls{tx} and \gls{rx} \cite{wu21}.
Thus, the received power writes as $P_R=P_T\left\vert\mathbf{h}_{R}\boldsymbol{\Theta}\mathbf{h}_{T}\right\vert^2$,
where $P_T=\mathbb{E}[\vert x\vert]$ is the transmitted power.
We consider the BD-RIS to be lossless, realistic for half-wavelength antenna spacing.

\begin{figure}[t]
\centering
\includegraphics[width=0.4\textwidth]{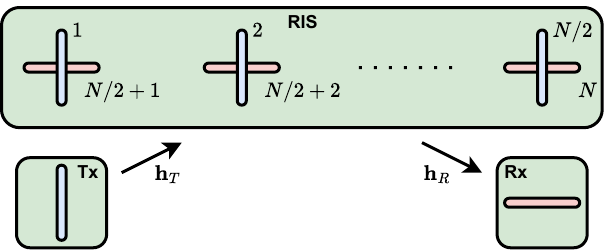}
\caption{Dual-polarized RIS-aided system where \gls{tx} and \gls{rx} have opposite polarization.}
\label{fig:system}
\end{figure}

We assume this RIS-aided system to be dual-polarized, where the RIS elements, the transmitting antenna, and the receiving antenna are either vertically or horizontally polarized.
Specifically, we assume $N/2$ RIS elements to have vertical polarization and $N/2$ RIS elements to have horizontal polarization, as commonly considered in dual-polarized systems \cite{kim10}.
With no loss of generality, we order the RIS elements such that the first $N/2$ are vertically polarized and the last $N/2$ are horizontally polarized, as shown in Fig.~\ref{fig:system}.
The effect of the antenna polarization is accounted for in the channels $\mathbf{h}_{R}$ and $\mathbf{h}_{T}$, by expressing them as $\mathbf{h}_{R}=\mathbf{p}_{R}\odot\tilde{\mathbf{h}}_{R}$, and $\mathbf{h}_{T}=\mathbf{p}_{T}\odot\tilde{\mathbf{h}}_{T}$, where $\odot$ denotes the Hadamard product, $\mathbf{p}_{R}\in\mathbb{R}^{1\times N}$ and $\mathbf{p}_{T}\in\mathbb{R}^{N\times 1}$ include the effects of the power imbalance between the different polarizations, and $\tilde{\mathbf{h}}_{R}\in\mathbb{C}^{1\times N}$ and $\tilde{\mathbf{h}}_{T}\in\mathbb{C}^{N\times 1}$ are the uni-polarized fading channels \cite{kim10}.
The vectors $\mathbf{p}_{R}$ and $\mathbf{p}_{T}$ are modeled depending on the polarizations of the \gls{tx} and \gls{rx}.
When the \gls{rx} is vertically or horizontally polarized, we have $\mathbf{p}_{R}=[1, \sqrt{\chi}]\otimes\mathbf{1}_{1\times N/2}$ or $\mathbf{p}_{R}=[\sqrt{\chi}, 1]\otimes\mathbf{1}_{1\times N/2}$, respectively, where $\otimes$ denotes the Kronecker product and $0\leq\chi\leq1$ is the inverse of the \gls{xpd} \cite{kim10}.
Similarly, when the \gls{tx} is vertically or horizontally polarized, we have
$\mathbf{p}_{T}=[1, \sqrt{\chi}]^T\otimes\mathbf{1}_{N/2\times 1}$ or $\mathbf{p}_{T}=[\sqrt{\chi}, 1]^T\otimes\mathbf{1}_{N/2\times 1}$, respectively.
The fading channels $\tilde{\mathbf{h}}_{R}$ and $\tilde{\mathbf{h}}_{T}$ are modeled in this study considering two popular fading distributions, namely \gls{iid} Rayleigh channels $\tilde{\mathbf{h}}_{R}\sim\mathcal{CN}(\mathbf{0},\mathbf{I})$ and $\tilde{\mathbf{h}}_{T}\sim\mathcal{CN}(\mathbf{0},\mathbf{I})$, and \gls{los} channels $\tilde{\mathbf{h}}_{R}=[e^{j\varphi_{R,1}},\ldots,e^{j\varphi_{R,N}}]$ and $\tilde{\mathbf{h}}_{T}=[e^{j\varphi_{T,1}},\ldots,e^{j\varphi_{T,N}}]^T$, where $\varphi_{R,n},\varphi_{T,n}\in[0,2\pi]$, for $n=1,\ldots,N$.

\begin{table*}[t]
\centering
\caption{Scaling laws of the received power and BD-RIS gain in uni- and dual-polarized RIS-aided systems.}
\begin{tabular}{c|c|cc|}
\cline{2-4}
& \multirow{2}{*}{Uni-polarized system} & \multicolumn{2}{c|}{Dual-polarized system}\\
\cline{3-4}
& & \multicolumn{1}{c|}{Tx and Rx have the same polarization} & Tx and Rx have opposite polarization\\
\hline
\multicolumn{1}{|c|}{Rayleigh} &
\begin{tabular}[c]{@{}c@{}}
$\bar{P}_R^{\mathrm{Single}}
=N+N\left(N-1\right)\frac{\pi^2}{16}$,\\
$\bar{P}_R^{\mathrm{Fully}}
=N^2$, $G=\frac{16}{\pi^2}$.
\end{tabular} &
\multicolumn{1}{c|}{\begin{tabular}[c]{@{}c@{}}
$\bar{P}_R^{\mathrm{Single}}
=\frac{1+\chi^2}{2}\left(N+N\left(\frac{N}{2}-1\right)\frac{\pi^2}{16}\right)+\frac{\pi^2\chi}{32}N^2$,\\
$\bar{P}_R^{\mathrm{Fully}}
=\frac{\left(1+\chi\right)^2}{4}N^2$, $G=\frac{16}{\pi^2}$.
\end{tabular}} &
\begin{tabular}[c]{@{}c@{}}
$\bar{P}_R^{\mathrm{Single}}
=\chi\left(N+N\left(N-1\right)\frac{\pi^2}{16}\right)$,\\
$\bar{P}_R^{\mathrm{Fully}}
=\frac{\left(1+\chi\right)^2}{4}N^2$, $G=\frac{4(1+\chi)^2}{\pi^2\chi}$.
\end{tabular}\\
\hline
\multicolumn{1}{|c|}{LoS} &
\begin{tabular}[c]{@{}c@{}}
$P_R^{\mathrm{Single}}=N^2$,\\
$P_R^{\mathrm{Fully}}=N^2$, $G=1$.
\end{tabular} &
\multicolumn{1}{c|}{\begin{tabular}[c]{@{}c@{}}
$P_R^{\mathrm{Single}}=\frac{\left(1+\chi\right)^2}{4}N^2$,\\
$P_R^{\mathrm{Fully}}=\frac{\left(1+\chi\right)^2}{4}N^2$, $G=1$.
\end{tabular}} & 
\begin{tabular}[c]{@{}c@{}}
$P_R^{\mathrm{Single}}
=\chi N^2$,\\
$P_R^{\mathrm{Fully}}
=\frac{\left(1+\chi\right)^2}{4}N^2$, $G=\frac{\left(1+\chi\right)^2}{4\chi}$.
\end{tabular}\\
\hline
\end{tabular}
\label{tab}
\end{table*}

\section{Received Power Scaling Laws and BD-RIS Gain}

In this section, we derive the scaling laws of the received power in the case of a D-RIS and a BD-RIS in a dual-polarized system.
Considering $P_T=1$ for simplicity, a D-RIS, also known as single-connected RIS \cite{she22}, achieves a received power
\begin{equation}
P_R^{\mathrm{Single}}=\left(\sum_{n=1}^N\left\vert\left[\mathbf{h}_{R}\right]_{n}\left[\mathbf{h}_{T}\right]_{n}\right\vert\right)^2,\label{eq:PR-single}
\end{equation}
while a BD-RIS can achieve a received power
\begin{equation}
P_{R}^{\mathrm{Fully}}=\left\Vert \mathbf{h}_{R}\right\Vert^2\left\Vert\mathbf{h}_{T}\right\Vert^{2},\label{eq:PR-fully}
\end{equation}
when the fully-connected RIS is used \cite{she22}.
We consider the received power achieved by a fully-connected RIS as this is the BD-RIS architecture enabling maximum performance.
Since the fully-connected RIS is more flexible than the single-connected RIS, we have $P_{R}^{\mathrm{Fully}}\geq P_R^{\mathrm{Single}}$, as it can be shown with the Cauchy–Schwarz inequality \cite{she22}.
In addition, we investigate the gain of BD-RIS over D-RIS, introduced as
\begin{equation}
G=\lim_{N\to+\infty}\frac{\bar{P}_R^{\mathrm{Fully}}}{\bar{P}_R^{\mathrm{Single}}},\label{eq:gain}
\end{equation}
where $\bar{P}_R^{\mathrm{Fully}}=\mathbb{E}[P_R^{\mathrm{Fully}}]$ and $\bar{P}_R^{\mathrm{Single}}=\mathbb{E}[P_R^{\mathrm{Single}}]$ are the received powers averaged across the channel realizations.

We study two scenarios, namely where the \gls{tx} and \gls{rx} have the same polarization and opposite polarization, as these are the two extreme cases that can occur due to \gls{tx} or \gls{rx} mobility.
For each of the two scenarios, we consider Rayleigh and \gls{los} fading channels, giving a total of four scenarios.

\subsection{Same Polarization at Tx and Rx with Rayleigh Channels}
\label{sec:same-Rayleigh}

Consider a dual-polarized RIS-aided system where the \gls{tx} and \gls{rx} have the same polarization (assumed to be vertical with no loss of generality), and the channels are \gls{iid} Rayleigh distributed.
For a single-connected RIS, since
$\vert[\mathbf{h}_{R}]_{n}[\mathbf{h}_{T}]_{n}\vert=\vert[\tilde{\mathbf{h}}_{R}]_{n}[\tilde{\mathbf{h}}_{T}]_{n}\vert$, for $n=1,\ldots,N/2$, and $\vert[\mathbf{h}_{R}]_{n}[\mathbf{h}_{T}]_{n}\vert=\chi\vert[\tilde{\mathbf{h}}_{R}]_{n}[\tilde{\mathbf{h}}_{T}]_{n}\vert$, for $n=N/2+1,\ldots,N$, \eqref{eq:PR-single} becomes
\begin{equation}
P_R^{\mathrm{Single}}
=\left(\sum_{n=1}^{N/2}\left\vert[\tilde{\mathbf{h}}_{R}]_{n}[\tilde{\mathbf{h}}_{T}]_{n}\right\vert
+\chi\sum_{n=N/2+1}^{N}\left\vert[\tilde{\mathbf{h}}_{R}]_{n}[\tilde{\mathbf{h}}_{T}]_{n}\right\vert\right)^2.\label{eq:single-same-ray0}
\end{equation}
Taking the expectation of \eqref{eq:single-same-ray0} by expanding the square and following the steps in \cite[Section~IV-E]{she22}, we get
\begin{equation}
\bar{P}_R^{\mathrm{Single}}=\frac{1+\chi^2}{2}\left(N+N\left(\frac{N}{2}-1\right)\frac{\pi^2}{16}\right)+\frac{\pi^2\chi}{32}N^2,\label{eq:single-same-ray}
\end{equation}
where we exploited $\mathbb{E}[\vert[\tilde{\mathbf{h}}_{R}]_{n}\vert]=\mathbb{E}[\vert[\tilde{\mathbf{h}}_{T}]_{n}\vert]=\sqrt{\pi}/2$ and $\mathbb{E}[\vert[\tilde{\mathbf{h}}_{R}]_{n}\vert^2]=\mathbb{E}[\vert[\tilde{\mathbf{h}}_{T}]_{n}\vert^2]=1$, for $n=1,\ldots,N$.
The received power $\bar{P}_R^{\mathrm{Single}}$ in \eqref{eq:single-same-ray} grows with $\chi$, spanning from $\bar{P}_R^{\mathrm{Single}}=N/2+N/2(N/2-1)\pi^2/16$ when $\chi=0$ to $\bar{P}_R^{\mathrm{Single}}=N+N(N-1)\pi^2/16$ when $\chi=1$.
Besides, for a fully-connected RIS, \eqref{eq:PR-fully} can be rewritten as
\begin{multline}
P_R^{\mathrm{Fully}}
=     \left(\left\Vert[\tilde{\mathbf{h}}_{R}]_{1:N/2}\right\Vert^2+\chi\left\Vert[\tilde{\mathbf{h}}_{R}]_{N/2+1:N}\right\Vert^2\right)\\
\times\left(\left\Vert[\tilde{\mathbf{h}}_{T}]_{1:N/2}\right\Vert^2+\chi\left\Vert[\tilde{\mathbf{h}}_{T}]_{N/2+1:N}\right\Vert^2\right),\label{eq:fully-same-ray0}
\end{multline}
since $\Vert\mathbf{h}_{R}\Vert^2=\Vert[\tilde{\mathbf{h}}_{R}]_{1:N/2}\Vert^2+\chi\Vert[\tilde{\mathbf{h}}_{R}]_{N/2+1:N}\Vert^2$ and $\Vert\mathbf{h}_{T}\Vert^2=\Vert[\tilde{\mathbf{h}}_{T}]_{1:N/2}\Vert^2+\chi\Vert[\tilde{\mathbf{h}}_{T}]_{N/2+1:N}\Vert^2$.
Thus, by exploiting the independence between $\tilde{\mathbf{h}}_{R}$ and $\tilde{\mathbf{h}}_{T}$, and that $\mathbb{E}[\Vert[\tilde{\mathbf{h}}_{R}]_{1:N/2}\Vert^2]=\mathbb{E}[\Vert[\tilde{\mathbf{h}}_{R}]_{N/2+1:N}\Vert^2]=\mathbb{E}[\Vert[\tilde{\mathbf{h}}_{T}]_{1:N/2}\Vert^2]=\mathbb{E}[\Vert[\tilde{\mathbf{h}}_{T}]_{N/2+1:N}\Vert^2]=N/2$, the expectation of \eqref{eq:fully-same-ray0} reads as
\begin{equation}
\bar{P}_R^{\mathrm{Fully}}=\left(\frac{N}{2}+\chi\frac{N}{2}\right)^2=\frac{\left(1+\chi\right)^2}{4}N^2,\label{eq:fully-same-ray}
\end{equation}
increasing with $\chi$.
Specifically, we have $\bar{P}_R^{\mathrm{Fully}}=(N/2)^2$ when $\chi=0$ and $\bar{P}_R^{\mathrm{Fully}}=N^2$ when $\chi=1$.
Note that \eqref{eq:single-same-ray} and \eqref{eq:fully-same-ray} with $\chi=1$ correspond to the scaling laws for uni-polarized systems \cite{she22}, while with $\chi=0$ they correspond to the same scaling laws of a RIS with $N/2$ elements, since only the $N/2$ vertically polarized elements can be effectively used.
By substituting \eqref{eq:single-same-ray} and \eqref{eq:fully-same-ray} into \eqref{eq:gain}, we obtain $G=16/\pi^2\approx1.62$.
The gain $G$ is independent of $\chi$ as it is computed for $N\rightarrow\infty$, and effectively using $N$ or $N/2$ RIS elements, in the extreme cases of $\chi=1$ and $\chi=0$, respectively, does not impact $G$.

\subsection{Same Polarization at Tx and Rx with \gls{los} Channels}

Consider a dual-polarized RIS-aided system where the \gls{tx} and \gls{rx} have the same polarization (assumed to be vertical), and the channels are \gls{los}.
With \gls{los} channels, we have $\vert[\mathbf{h}_{R}]_{n}[\mathbf{h}_{T}]_{n}\vert=1$, for $n=1,\ldots,N/2$, $\vert[\mathbf{h}_{R}]_{n}[\mathbf{h}_{T}]_{n}\vert=\chi$, for $n=N/2+1,\ldots,N$, and $\Vert\mathbf{h}_{R}\Vert^2=\Vert\mathbf{h}_{T}\Vert^2=N/2+\chi N/2$.
Thus, both \eqref{eq:PR-single} and \eqref{eq:PR-fully} boil down to
\begin{equation}
P_R^{\mathrm{Single}}=P_R^{\mathrm{Fully}}=\frac{(1+\chi)^2}{4}N^2,\label{eq:same-los}
\end{equation}
giving the received power achieved by both single- and fully-connected RIS.
Consequently, the gain of BD-RIS over D-RIS is $G=1$, independently of the value of $\chi$.

\subsection{Opposite Polarization at Tx and Rx with Rayleigh Channels}

Consider the \gls{tx} and \gls{rx} to have opposite polarization (vertically and horizontally polarized, respectively, as in Fig.~\ref{fig:system}, with no loss of generality), and Rayleigh distributed channels.
By noticing that $\vert[\mathbf{h}_{R}]_{n}[\mathbf{h}_{T}]_{n}\vert=\sqrt{\chi}\vert[\tilde{\mathbf{h}}_{R}]_{n}[\tilde{\mathbf{h}}_{T}]_{n}\vert$, for $n=1,\ldots,N$, \eqref{eq:PR-single} can be rewritten as
\begin{equation}
P_R^{\mathrm{Single}}=\chi\left(\sum_{n=1}^N\left\vert[\tilde{\mathbf{h}}_{R}]_{n}[\tilde{\mathbf{h}}_{T}]_{n}\right\vert\right)^2,
\end{equation}
whose expectation is given by
\begin{equation}
\bar{P}_R^{\mathrm{Single}}=\chi\left(N+N\left(N-1\right)\frac{\pi^2}{16}\right),\label{eq:single-oppo-ray}
\end{equation}
following the steps in \cite[Section~IV-E]{she22}.
Clearly, the scaling law in \eqref{eq:single-oppo-ray} increases with $\chi$, being $\bar{P}_R^{\mathrm{Single}}=0$ when $\chi=0$ and $\bar{P}_R^{\mathrm{Single}}=N+N(N-1)\pi^2/16$ when $\chi=1$.
Besides, the fully-connected RIS achieves a received power of
\begin{equation}
\bar{P}_R^{\mathrm{Fully}}=\frac{\left(1+\chi\right)^2}{4}N^2,\label{eq:fully-oppo-ray}
\end{equation}
following the discussion in Section~\ref{sec:same-Rayleigh}.
The gain of BD-RIS over D-RIS is given by substituting \eqref{eq:single-oppo-ray} and \eqref{eq:fully-oppo-ray} into \eqref{eq:gain} as
\begin{equation}
G=\frac{4\left(1+\chi\right)^2}{\pi^2\chi},\label{eq:gain-oppo-ray}
\end{equation}
which decreases with $\chi$.
Interestingly, we have $G=+\infty$ when $\chi=0$ and $G=16/\pi^2\approx1.62$ when $\chi=1$, showing that the gain is significant for small values of $\chi$.

\subsection{Opposite Polarization at Tx and Rx with LoS Channels}

Consider the \gls{tx} and \gls{rx} to have opposite polarization (vertically and horizontally polarized, respectively), and the channels to be \gls{los}.
With \gls{los} channels, $\vert[\mathbf{h}_{R}]_{n}[\mathbf{h}_{T}]_{n}\vert=\sqrt{\chi}$, for $n=1,\ldots,N$.
Thus, \eqref{eq:PR-single} simplifies to
\begin{equation}
P_R^{\mathrm{Single}}
=\left(\sum_{n=1}^N\sqrt{\chi}\right)^2
=\chi N^2,\label{eq:single-oppo-los}
\end{equation}
increasing with $\chi$ and ranging from $P_R^{\mathrm{Single}}=0$ to $P_R^{\mathrm{Single}}=N^2$.
For a fully-connected RIS, \eqref{eq:PR-fully} boils down to
\begin{equation}
P_R^{\mathrm{Fully}}=\frac{\left(1+\chi\right)^2}{4}N^2,\label{eq:fully-oppo-los}
\end{equation}
since $\Vert\mathbf{h}_{R}\Vert^2=\Vert\mathbf{h}_{T}\Vert^2=N/2+\chi N/2$.
By substituting \eqref{eq:single-oppo-los} and \eqref{eq:fully-oppo-los} into \eqref{eq:gain}, we obtain
\begin{equation}
G=\frac{\left(1+\chi\right)^2}{4\chi},\label{eq:gain-oppo-los}
\end{equation}
decreasing with $\chi$.
We observe that $G=+\infty$ when $\chi=0$ and $G=1$ when $\chi=1$, showing remarkable gains of BD-RIS over D-RIS also with \gls{los} channels, particularly with small $\chi$.


We summarize our analytical findings for the four scenarios in Tab.~\ref{tab}, where we compare them with the performance in the case of a uni-polarized system provided in \cite{she22}.
We report the gain of BD-RIS over D-RIS $G$ in the four considered scenarios as a function of $\chi$ in Fig.~\ref{fig:gain}.
Interestingly, BD-RIS can offer a gain $G>1$ over D-RIS also with \gls{los} channels when considering a dual-polarized system.
This gain is particularly high for small values of $\chi$.
For example, we have $G=3$ with \gls{los} channels when the \gls{tx} and \gls{rx} have opposite polarization and $\chi=0.1$.
Additional numerical simulations show that the gain $G$ under Rician channels increases as the Rician factor decreases and always remains between the gains under LoS and Rayleigh channels shown in Fig.~\ref{fig:gain}.

\begin{figure}[t]
\centering
\includegraphics[width=0.36\textwidth]{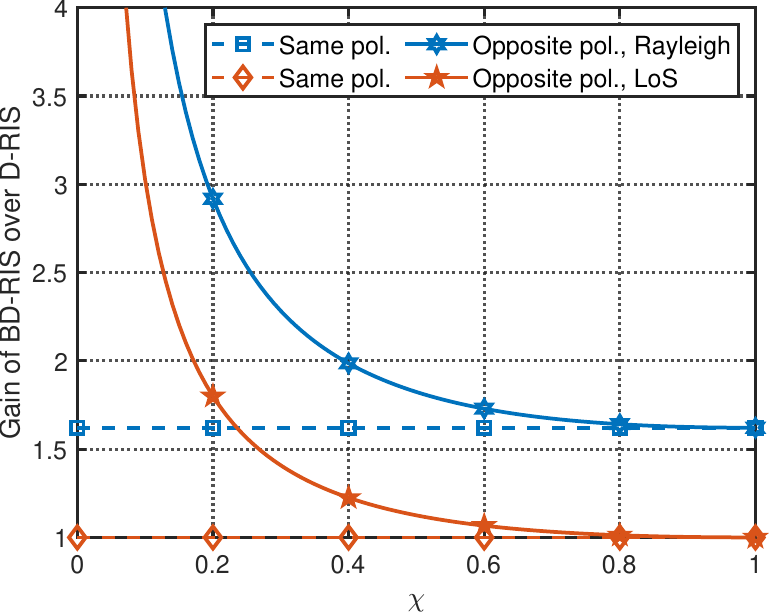}
\caption{Gain of BD-RIS over D-RIS $G$ as a function of $\chi$, when the \gls{tx} and \gls{rx} have the same/opposite polarization, and with Rayleigh/\gls{los} channels.}
\label{fig:gain}
\end{figure}

\section{Optimal BD-RISs with LoS Channels}

We have shown that BD-RIS is beneficial over D-RIS also under \gls{los} channels, in the case of a dual-polarized wireless system where the \gls{tx} and \gls{rx} have opposite polarization.
Focusing on this scenario, unexplored in previous work, we characterize the optimal BD-RIS architectures that achieve the Pareto frontier of the performance-complexity trade-off, where the performance is the received power $P_R$ achievable by the BD-RIS and the complexity as the number of tunable impedance components $C$ in the BD-RIS.
We begin by introducing a low-complex BD-RIS architecture that can achieve the same performance as the fully-connected RIS.

\begin{proposition}
Consider a dual-polarized RIS-aided system with \gls{tx} and \gls{rx} having opposite polarization and \gls{los} channels.
In this system, any group-connected RIS with group size 2 where each group contains two RIS elements with opposite polarization achieves the performance upper bound $P_R^{\mathrm{Fully}}$ in \eqref{eq:fully-oppo-los} for any value of $\chi$.
\label{pro:1}
\end{proposition}

\begin{proof}
Following \cite{she22}, the received power achieved by a group connected RIS with group size 2 is given by
\begin{equation}
P_R^{\mathrm{Group2}}
=\left(\sum_{g=1}^{N/2}\left\Vert\mathbf{h}_{R,g}\right\Vert\left\Vert\mathbf{h}_{T,g}\right\Vert\right)^2,\label{eq:PR-group2}
\end{equation}
where $\mathbf{h}_{R,g}\in\mathbb{C}^{1\times 2}$ and $\mathbf{h}_{T,g}\in\mathbb{C}^{2\times 1}$ contain the two entries of $\mathbf{h}_{R}$ and $\mathbf{h}_{T}$, respectively, corresponding to the two RIS elements in the $g$th group.
Assume with no loss of generality that the \gls{tx} is vertically polarized and the \gls{rx} is horizontally polarized, as in Fig.~\ref{fig:system}.
Since each group contains two RIS elements with opposite polarization, $\mathbf{h}_{R,g}=[\sqrt{\chi}e^{j\varphi_{R,n_g}},e^{j\varphi_{R,N/2+{m_g}}}]$ and $\mathbf{h}_{T,g}=[e^{j\varphi_{T,n_g}},\sqrt{\chi}e^{j\varphi_{T,N/2+{m_g}}}]^T$, for some $n_g,m_g\in\{1,\ldots,N/2\}$.
Thus, we have $\Vert\mathbf{h}_{R,g}\Vert=\Vert\mathbf{h}_{T,g}\Vert=\sqrt{1+\chi}$, allowing us to rewrite \eqref{eq:PR-group2} as
\begin{equation}
P_R^{\mathrm{Group2}}
=\left(\sum_{g=1}^{N/2}\left(1+\chi\right)\right)^2
=\frac{\left(1+\chi\right)^2}{4}N^2,
\end{equation}
which is equal to \eqref{eq:fully-oppo-los}, and proves the proposition.
\end{proof}

Proposition~\ref{pro:1} implies that the performance upper bound $P_R^{\mathrm{Fully}}$ in \eqref{eq:fully-oppo-los} can be achieved with a BD-RIS architecture having circuit complexity $C=3N/2$ ($N$ impedance components connecting each of the $N$ RIS elements to ground, and an impedance component in each of the $N/2$ groups).
Remarkably, a complexity of $C=3N/2$ is considerably less than the complexity of the fully-connected RIS, being $C=N(N+1)/2$ \cite{she22}.
To bridge between the single-connected RIS ($C=N$) and the group-connected RIS with group size 2 ($C=3N/2$), the following proposition provides the maximum performance that can be achieved with a BD-RIS with complexity $C=N+n$, for $n=1,\ldots,N/2-1$, and characterizes the BD-RIS architecture achieving such performance, denoted as ``optimal''.

\begin{proposition}
The optimal BD-RIS architecture with $C=N+n$ tunable impedance components, for some $n\in\{1,\ldots,N/2-1\}$, has $n$ groups with size 2 each containing two RIS elements with opposite polarization and $N-2n$ groups with size 1.
This optimal BD-RIS architecture achieves a received power given by
\begin{equation}
P_R^{(n)}=\left(n\left(1+\chi\right)+\left(N-2n\right)\sqrt{\chi}\right)^2.
\end{equation}
\label{pro:2}
\end{proposition}

\begin{proof}
An optimal BD-RIS architecture with $N$ elements and complexity $C=N+n$, for some $n\in\{1,\ldots,N/2-1\}$, has $G=2N-C=N-n$ groups, according to \cite[Proposition~1]{ner23}.
Thus, such a BD-RIS achieves a received power
\begin{equation}
P_R^{(n)}
=\left(\sum_{g=1}^{N-n}\left\Vert\mathbf{h}_{R,g}\right\Vert\left\Vert\mathbf{h}_{T,g}\right\Vert\right)^2,\label{eq:PR-n1}
\end{equation}
where $\mathbf{h}_{R,g}\in\mathbb{C}^{1\times N_g}$ and $\mathbf{h}_{T,g}\in\mathbb{C}^{N_g\times 1}$ contain the $N_g$ entries of $\mathbf{h}_{R}$ and $\mathbf{h}_{T}$, respectively, corresponding to the $N_g$ RIS elements in the $g$th group, and $N_g$ denotes the group size of the $g$th group.
Since the sum of the sizes of the $N-n$ groups must be $N$, i.e., $\sum_{g=1}^{N-n}N_g=N$, we have at least $N-2n$ groups with size 1.
This can be proved by contradiction: assuming fewer than $N-2n$ values $N_g$ are 1, specifically $N-2n-k$ with $k\geq1$, the remaining $n+k$ values $N_g$ must sum to $2n+k$, which is impossible because the minimum sum of $n+k$ natural numbers greater than 1 is $2n+2k$.
Since there are at least $N-2n$ groups with size 1, we can rewrite \eqref{eq:PR-n1} as
\begin{equation}
P_R^{(n)}
=\left(\sum_{g=1}^{n}\left\Vert\mathbf{h}_{R,g}\right\Vert\left\Vert\mathbf{h}_{T,g}\right\Vert+\sum_{g=n+1}^{N-n}\left\vert h_{R,g}h_{T,g}\right\vert\right)^2,\label{eq:PR-n2}
\end{equation}
where we assumed with no loss of generality that $N-2n$ groups with size 1 are indexed $g=n+1,\ldots,N-n$.
By noticing that $\vert h_{R,g}h_{T,g}\vert=\sqrt{\chi}$, for $g=n+1,\ldots,N-n$, since \gls{tx} and \gls{rx} have opposite polarization, \eqref{eq:PR-n2} simplifies as
\begin{equation}
P_R^{(n)}
=\left(\sum_{g=1}^{n}\left\Vert\mathbf{h}_{R,g}\right\Vert\left\Vert\mathbf{h}_{T,g}\right\Vert+\left(N-2n\right)\sqrt{\chi}\right)^2,\label{eq:PR-n3}
\end{equation}
and the problem of maximizing $P_R^{(n)}$ becomes equivalent to maximizing the term
\begin{equation}
Q=\sum_{g=1}^{n}\left\Vert\mathbf{h}_{R,g}\right\Vert\left\Vert\mathbf{h}_{T,g}\right\Vert,\label{eq:Q}
\end{equation}
by designing the polarizations of the $2n$ RIS elements considered in the term $Q$ and how to divide them into $n$ groups.
Specifically, among the $2n$ RIS elements in $Q$, $n_v$ have vertical polarization and $n_h$ have horizontal polarization, such that $n_v+n_h=2n$, and $n_v$ and $n_h$ need to be optimized to maximize $Q$.
In the following, we maximize $Q$ by assuming that $n_v=n_h=n$, and we show afterward that this is the optimal design of the polarizations.

Assuming that $n_v=n_h=n$, the term $Q$ is maximized when the $n$ groups have size 2 and each contains two RIS elements with opposite polarization, following Proposition~\ref{pro:1}.
In this case, we have
\begin{equation}
Q=\sqrt{\frac{\left(1+\chi\right)^2}{4}\left(2n\right)^2}=n\left(1+\chi\right),\label{eq:Q2}
\end{equation}
because of Proposition~\ref{pro:1}.
Thus, by substituting \eqref{eq:Q2} into \eqref{eq:PR-n3}, we obtain that the optimal received power is given by
\begin{equation}
P_R^{(n)}
=\left(n\left(1+\chi\right)+\left(N-2n\right)\sqrt{\chi}\right)^2,\label{eq:PR-n4}
\end{equation}
which proves the proposition if $n_v=n_h=n$ is optimal.

\begin{figure}[t]
\centering
\includegraphics[width=0.36\textwidth]{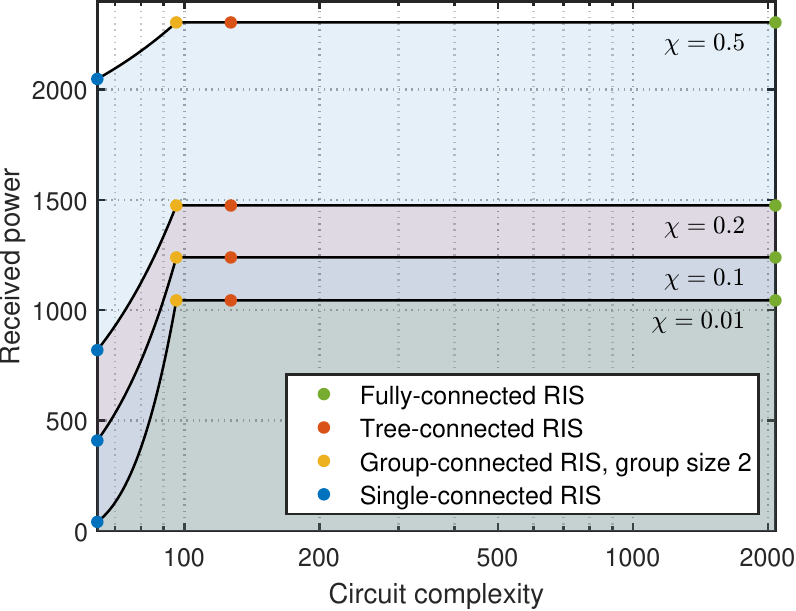}
\caption{Pareto frontier of the performance-complexity trade-off achieved by dual-polarized BD-RISs, with $N=64$.}
\label{fig:Pareto}
\end{figure}

To prove that $n_v=n_h=n$ is optimal to maximize $Q$, we upper bound $Q$ by
\begin{equation}
Q\leq\left\Vert\hat{\mathbf{h}}_{R}\right\Vert\left\Vert\hat{\mathbf{h}}_{T}\right\Vert,\label{eq:Q3}
\end{equation}
where $\hat{\mathbf{h}}_{R}\in\mathbb{C}^{1\times2n}$ and $\hat{\mathbf{h}}_{T}\in\mathbb{C}^{2n\times1}$ are introduced as $\hat{\mathbf{h}}_{R}=[\mathbf{h}_{R,1},\ldots,\mathbf{h}_{R,n}]$ and $\hat{\mathbf{h}}_{T}=[\mathbf{h}_{T,1}^T,\ldots,\mathbf{h}_{T,n}^T]^T$, which follows from the Cauchy–Schwarz inequality \cite{she22}.
Assuming with no loss of generality that the \gls{tx} is vertically polarized and the \gls{rx} is horizontally polarized, we have $\Vert\hat{\mathbf{h}}_{R}\Vert=\sqrt{n_v\chi+n_h}$ and $\Vert\hat{\mathbf{h}}_{R}\Vert=\sqrt{n_v+n_h\chi}$, and we can express \eqref{eq:Q3} as
\begin{align}
Q&\leq\sqrt{\left(n_v\chi+n_h\right)\left(n_v+n_h\chi\right)}\\
&=\sqrt{-\left(1-\chi\right)^2n_{h}^2+2n\left(1-\chi\right)n_{h}+4n^2\chi},\label{eq:Q4}
\end{align}
where we used $n_v+n_h=2n$.
By taking the derivative of the upper bound in \eqref{eq:Q4}, we find that it is maximized when $n_h=n$, yielding a maximum value of
\begin{equation}
Q\leq\sqrt{n^2\left(1-\chi\right)^2+4n^2\chi}=n\left(1+\chi\right),
\end{equation}
which is the same as \eqref{eq:Q2} proving the optimality of the solution with $n_v=n_h=n$ and concluding the proof.
\end{proof}

In Fig.~\ref{fig:Pareto}, we report the Pareto frontier of the performance-complexity trade-off offered by BD-RIS, by fixing $N=64$, when \gls{tx} and \gls{rx} have opposite polarization and with \gls{los} channels, as given in Proposition~\ref{pro:2}.
We make the following two observations.
\textit{First}, the single-connected RIS is the simplest architecture achieving the lowest performance and the fully-connected RIS is the most complex architecture achieving the highest performance.
The two architectures give different performances with \gls{los} channels as long as $\chi\neq1$, and the performance gap increases as $\chi$ decreases.
Interestingly, the gain of BD-RIS over D-RIS under \gls{los} remained unexplored in previous work.
\textit{Second}, the group-connected RIS with group size 2, is the least complex architecture achieving the performance upper bound, i.e., the same performance as the fully-connected RIS.
Thus, fully-/tree-connected RISs and group-/forest-connected RISs with group sizes different from 2 achieve the performance upper bound, but with increased circuit complexity and the tree-connected RIS is not the least complex BD-RIS achieving maximum performance in this scenario.

\section{Conclusion}

We analyze the fundamental performance limits of dual-polarized BD-RIS, by deriving the scaling laws of the achievable received power and the gain offered over D-RIS.
We observe that BD-RIS offers important gains over D-RIS also with \gls{los} channels, especially when the inverse of the \gls{xpd} is small.
In addition, we derive the Pareto frontier of the performance-complexity trade-off achieved by dual-polarized BD-RISs, showing that the received power upper bound is achieved with a group-connected RIS with group size 2.
This new insight is crucial for guiding the future prototyping and deployment of low-complexity BD-RIS architectures.

\bibliographystyle{IEEEtran}
\bibliography{IEEEabrv,main}

\end{document}